\documentclass[letterpaper, peerreview, 12pt, draftclsnofoot]{IEEEtran}

\usepackage{graphicx}

\newcommand{\N}{\mathbb{N}}  % naturlige tal
\newcommand{\C}{\mathbb{C}}  % komplekse tal
\newcommand{\F}{\mathbb{F}}  % galois field 
\usepackage{fancyhdr}
\usepackage[T1]{fontenc}
\usepackage{physics}
\usepackage{cite}
\usepackage{url}
\usepackage{algorithm}
\usepackage{algpseudocode}
\usepackage{balance}                
\usepackage{dirtytalk} 
\usepackage{amsmath,amssymb,amsfonts,stmaryrd}		
\usepackage{mathtools}					
\mathtoolsset{showonlyrefs, showmanualtags}
\usepackage{textcomp}                 		
\usepackage{siunitx}						
\usepackage[official]{eurosym} 

\usepackage{soul}
\usepackage{bm}
\makeatletter
\newcommand*{\bigs}[1]{{\hbox{$\left#1\vbox to10\p@{}\right.\n@space$}}}
\makeatother
\makeatletter
\newcommand*{\biggs}[1]{{\hbox{$\left#1\vbox to17\p@{}\right.\n@space$}}}
\makeatother
\usepackage[hidelinks]{hyperref}
\usepackage[
    protrusion=true,
    activate={true,nocompatibility},
    final,
    tracking=true,
    kerning=true,
    spacing=true,
    factor=1100]{microtype}
\SetTracking{encoding={*}, shape=sc}{40}
\usepackage[english]{babel}
\usepackage{blindtext}
\usepackage{tikz}  
\usepackage{nicematrix}
\usetikzlibrary{calc}
\usetikzlibrary{decorations.pathreplacing,calligraphy}
\usepackage[framemethod=tikz]{mdframed}
\usepackage{pifont} 
\makeatletter
\newcommand*\bigcdot{\mathpalette\bigcdot@{.5}}
\newcommand*\bigcdot@[2]{\mathbin{\vcenter{\hbox{\scalebox{#2}{$\m@th#1\bullet$}}}}}
\makeatother
\usetikzlibrary{shapes.geometric, arrows, arrows.meta}
\usepackage{subcaption}
\usepackage{graphicx}
\usepackage{float}
\usepackage{booktabs}
\usepackage{multirow}
\usepackage{dblfloatfix}
\usepackage[shortlabels]{enumitem}
\DeclareSIUnit[quantity-product = ]\percent{\char`\%}
\usepackage{comment}
\usetikzlibrary{decorations.text}
\usetikzlibrary{decorations.pathreplacing,calligraphy}
\usepackage[framemethod=tikz]{mdframed}
\usetikzlibrary{shapes.geometric, arrows}
\usetikzlibrary{positioning}

\usepackage{ntheorem}
\theoremheaderfont{\normalfont\bfseries}
\theorembodyfont{\itshape}
\theoremsymbol{}
\theoremstyle{plain}
\newtheorem{definition}{Definition}

\usepackage[capitalise]{cleveref}
\crefname{figure}{figure}{figures}
\Crefname{figure}{Figure}{Figures}
\usepackage{bbm}
\newtheorem{theorem}{Theorem}

\newenvironment{proof}{%
    \par\noindent\textit{Proof: }\ignorespaces
}{%
    \hfill$\blacksquare$\par
}
\makeatletter

\usepackage{arydshln}

\title{Puncturing for Adaptive Entanglement-Assisted Stabilizer Codes}
\author{Nicolai Peder Bülow Pedersen, Jakob Kaltoft Søndergaard, Jaron Skovsted Gundersen, René Bødker Christensen, Petar Popovski
\thanks{N.P.B. Pedersen (npbp@es.aau.dk), J.K. Søndergaard (jakobks@es.aau.dk), J.S. Gundersen (jaron@es.aau.dk), and P. Popovski (petarp@es.aau.dk) are with the Department of Electronic Systems at Aalborg University.}
\thanks{R.B. Christensen (rene@math.aau.dk) is with the Department of Mathematical Sciences at Aalborg University.}
}

\date{}

\begin{document}

\maketitle
\pagestyle{plain}
\thispagestyle{empty}

\begin{abstract}
    Quantum error-correcting codes are essential for reliable quantum communication. The achievable parameters can be improved by using entanglement-assisted stabilizers at the cost of consuming pre-shared Bell pairs. If the availability of such Bell pairs is intermittent, the ability to adapt the code to available resources is beneficial. That is, to pick a single mother code that can be adapted to different entanglement budgets. To achieve this, we extend an existing puncturing construction for standard/unassisted stabilizers to the entanglement-assisted setting through a three-step procedure: Interpreting the code as an unassisted stabilizer, puncturing a receiver-side qubit, and recovering an entanglement-assisted representation. This procedure reduces the number of required Bell pairs by one while preserving the number of transmitted and logical qubits. 
    We derive a bound on the resulting distance loss and establish conditions for preserving the original code distance. Finally, we analyse distinct, randomly generated, entanglement-assisted stabilizers with five to eight transmitted qubits. Of these codes, $14.75\%$ have a greater distance than their associated unassisted representations. Among this subset, $81.4\%$ admit at least one puncturing choice that preserves the original entanglement-assisted distance. These results illustrate how puncturing can reduce entanglement requirements while maintaining code distance. 
\end{abstract}


\section{\textbf{Introduction}}\label{sec_introduction}
Quantum error-correction (QEC) is essential for reliable quantum information processing due to the susceptibility of quantum states to noise~\cite{preskill1998reliable ,gottesman2010introduction, shor1995scheme}. Quantum information is inevitably affected by errors arising from imperfect quantum operations, transmissions through noisy channels, and storage in quantum memories. Stabilizer codes constitute a widely studied class of quantum error-correcting codes (QECCs) and have a close connection to classical linear codes~\cite{SC&QEC}. For instance, self-orthogonal classical linear codes can be used to construct stabilizer codes using the CSS construction~\cite{calderbank1996good, steane1996multiple}.

Entanglement-assisted quantum error-correcting codes (EAQECCs) extend the framework of stabilizer codes by exploiting pre-shared Bell pairs between the sender and receiver~\cite{CQEwE, hsieh2007general}. The presence of pre-shared entanglement removes the self-orthogonality requirement, allowing \emph{any} classical linear code to be used for construction of QECCs. Furthermore, EAQECCs can achieve lower code lengths for a fixed dimension and distance compared to stabilizer codes by investing entanglement~\cite{brun2014catalytic}. An EAQECC that encodes $k$ logical qubits into $n$ physical qubits while consuming $c$ pre-shared Bell pairs is denoted an $[\![n,k;c]\!]$ code. 

The operation of an entanglement-assisted communication system consists of a setup phase and a communication phase, as illustrated in \Cref{fig_com_framework}. During setup, a source distributed Bell pairs between the transmitter and receiver, each of which stores on half of every pair. During communication, the transmitter jointly encodes the logical information, locally prepared ancilla qubits, and its halves of the shared Bell pairs. The encoded qubits are then transmitted through a noisy channel. The receiver performs error correction and decoding using both the received qubits and its locally stored halves of the Bell pairs. The setup phase can be performed in advance of communication, allowing entanglement resources to be accumulated before the logical information is available for transmission. However, these resources must retain sufficient fidelity throughout the relevant storage intervals. In particular, the receiver-side qubits must remain coherent until decoding, consistent with the noiseless receiver assumption underlying the standard EAQECC framework. Entanglement purification~\cite{bennett1996purification} can improve the fidelity of the shared Bell pairs, while memory-protection techniques can mitigate errors during storage \cite{eaqecc_imperfect_ebits}. 

\begin{figure*}[h]
    \centering
    \includegraphics[width=\textwidth]{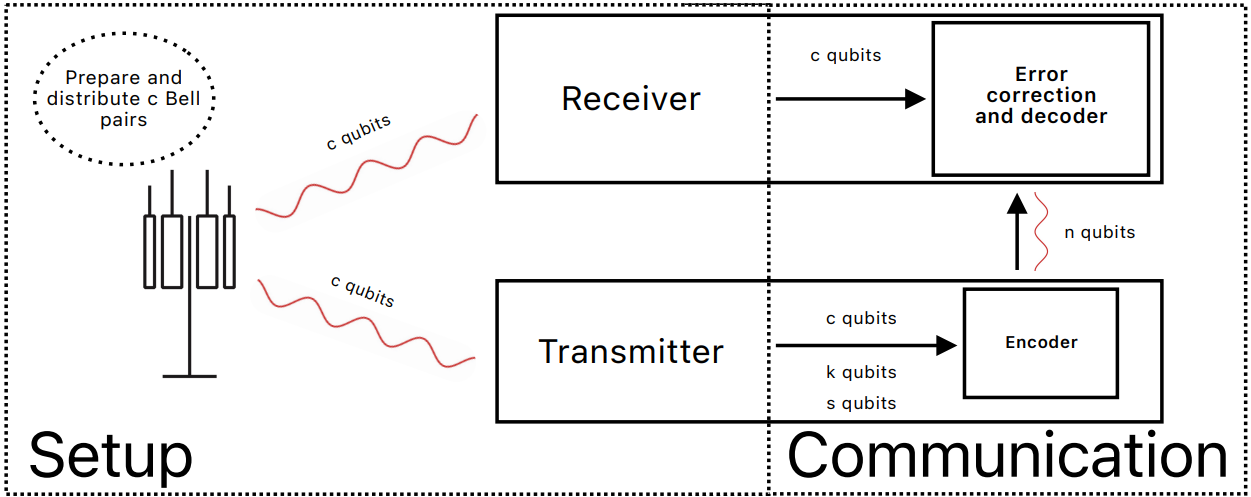}
    \caption{Communication framework for an entanglement-assisted stabilizer code. A base station distributes $c$ Bell pairs between the transmitter and receiver. The transmitter applies the encoding operation jointly to the logical qubits, ancilla qubits, and its halves of the shared Bell pairs. The resulting encoded qubits are transmitted through a noisy quantum channel. The receiver performs error correction and decoding using both the received qubits and its locally stored halves of the Bell pairs.}
    \label{fig_com_framework}
\end{figure*}

The availability of shared entanglement is therefore subject to physical constraints, including finite generation rates, limited quantum memory capacity, and decoherence during storage ~\cite{wehner2018quantum}. Moreover, entanglement is consumed during communication, so the number of Bell pairs available at a given transmission instance depends on the history of generation, consumption, storage, and loss. The available entanglement may consequently be insufficient for the originally chosen code. This motivates a resource-aware approach to EAQECCs, in which the code is adapted to the available entanglement.

Rather than considering each entanglement budget as a separate code-design problem, we seek a construction where a single base EAQECC can be adopted to a new EAQECC requiring no more entanglement than what is currently available. This provides a systematic approach where the base code determines a family of resource-adapted codes. While this approach does not guarantee optimal code construction, it provides a structural alternative to designing separate codes for each entanglement budget when necessary.

Puncturing of codes has long been an important tool for rate optimization in both classical and quantum coding theory~\cite{hagenauer1988rate, punct_gf4}. The code rate can be improved by removing parity bits/qubits after encoding, however, possibly at the cost of reduced error-correcting capabilities. In quantum coding theory, constructions with punctured codes were first established for qubit codes in~\cite{punct_gf4}. Recently,~\cite{punct_qudits} proposed code construction using puncturing of quantum qudit codes, while~\cite{PQSC} specialized the construction to stabilizer codes. These works characterize how the code parameters change under puncturing.

In this paper, we extend the puncturing construction of~\cite{PQSC} from stabilizer codes to EAQECCs. The proposed puncturing is resource-aware in the sense that puncturing is performed on the receiver-side qubit of Bell pairs such that the punctured code does not require more entanglement than currently available. This is fundamentally different than puncturing stabilizer codes, which is done to adapt to the conditions in which the code is used. A similar resource-aware puncturing framework was recently proposed for stabilizer codes used for encoded teleportation~\cite{abouamer2026latency}. This work develops a framework to puncture an EAQECC to adapt to the available resources that is consumed by the code. We develop a puncturing algorithm that systematically removes the receiver-side Bell-pair qubits while maintaining a valid EAQECC. This provides a systematic construction of a family of EAQECCs that can accommodate different resource availability scenarios obtained from a single base code. Extending the puncturing algorithm from \cite{PQSC} to encompass EAQECCs requires accounting for the different roles of the transmitter- and receiver-side qubits. Receiver-side qubits are assumed noiseless and provide an entanglement resource that supports error correction. Removing one can therefore weaken the protection of the transmitted information, even though no transmitted qubit is removed. For example, an EAQECC capable of correcting any single qubits error on the transmitted qubits may, after a single receiver-side puncture, become unable even to detect certain single-qubit errors. Although the existing stabilizer puncturing construction can be applied to the joint sender-receiver code, its usual distance guarantee concerns errors on the entire joint system and does not directly characterize this loss of protection. Moreover, the puncturing algorithm of~\cite{PQSC} transforms a stabilizer code on $n$ qubits into one on $n-1$ qubits. In the entanglement-assisted setting, a code using $n$ transmitted qubits and $c$ pre-shared Bell pairs has a joint transmitter--receiver stabilizer on $n+c$ qubits. Applying the existing construction to a receiver-side qubit therefore yields a stabilizer on $n+c-1$ qubits. However, this total qubit count alone does not establish that the resulting code admits an entanglement-assisted implementation using the same $n$ transmitted qubits and only $c-1$ pre-shared Bell pairs while preserving the number of logical qubits. Establishing this correspondence and analysing the resulting entanglement-assisted distance are therefore essential to extending the puncturing construction.

The remainder of this paper is organized as follows. \Cref{sec_preliminaries} introduces the necessary preliminaries, including stabilizer codes and an existing puncturing technique for stabilizer codes. \Cref{sec_eas} presents the construction of entanglement-assisted stabilizers. In \Cref{sec_punct_eas}, the puncturing framework is extended to entanglement-assisted stabilizers. We establish conditions under which the minimum distance is preserved under puncturing and provide a practical evaluation comparing the distance properties of entanglement-assisted stabilizers with their unassisted counterparts, including whether these properties are maintained under puncturing. \Cref{sec_discussion} discusses the assumptions underlying the results. Finally, \Cref{sec_conclusion} summarizes the main findings, and \Cref{sec_future_work} outlines open problems and directions for future research.
\section[Methodology]{\textbf{Preliminaries}}\label{sec_preliminaries}

Throughout this work, $\F_2$ denotes the Galois field of order $2$. Scalars in $\F_2$ are denoted by lowercase letters, while vectors in $\F_2^n$ are denoted by bold lowercase letters. For, $\mathbf{a},\mathbf{b}\in\F_2^n$, $[\mathbf{a}\ \mathbf{b}]\in\F_2^{2n}$ denotes their concatenation (the same notation may be used for scalars). Furthermore, we let $[a;b]:=\{i\}_{i=a}^{b}$ for $a,b\in\N$ with $a<b$. An $n$-qubit Pauli operator is on the form $P=P_1\otimes\cdots\otimes P_n\in(\C^2)^{\otimes n}$ with $P_i\in\{I,X,Z,Y\}$. We write Pauli operators as strings by omitting $\otimes$ and qubits that are affected trivially. For instance, $P=I\otimes X\otimes I\otimes Z$ is denoted $P=X_2Z_4$. The $n$-qubit Pauli group is defined as $\widetilde{\mathcal{P}}_n=\{\omega P\mid\omega\in\{\pm1,\pm i\},P\in\{I,X,Z,Y\}^{\otimes n}\}$. The corresponding quotient $\mathcal{P}_n=\widetilde{\mathcal{P}}_n/\langle\omega I\rangle$ is then the $n$-qubit Pauli group modulo the phase $\omega$. We denote vector $i$ in a set of vectors by superscript $(i)$ and entry $j$ of a vector by subscript $j$, thus $\mathbf{a}^{(i)}_j$ denotes entry $j$ of vector $i$ in the set $\{\mathbf{a}^{(l)}\}_l$.

\subsection{Stabilizer Codes}\label{subsec_stabilizers}
We briefly introduce the necessary theory for describing stabilizer codes\cite{QC&QI, SC&QEC}. A stabilizer is an Abelian subgroup $S\leq\widetilde{\mathcal{P}}_n$ satisfying $-I\notin S$. Suppose that $\{s_1,\ldots,s_{n-k}\}$ is a minimal generating set for $S$, denoted $S=\langle s_1,\ldots,s_{n-k}\rangle$. The associated codespace, $V_S$, of $S$ is then defined as the simultaneous $+1$ eigenspace of all $s\in S$:
\begin{align*}
    V_S := \big\{\ket{\psi}\mid s\ket{\psi}=\ket{\psi},\ \forall s\in S\big\}\subset(\mathbb{C}^{2})^{\otimes n}.
\end{align*}
Since the generating set is independent, $V_S$ is a $2^k$-dimensional subspace of $(\C^2)^{\otimes n}$.
Consequently, the stabilizer encodes $k$ logical qubits into $n$ physical qubits.

To characterize the error-correcting capabilities of a stabilizer, it suffices to consider phase-free Pauli errors (elements of $\mathcal{P}_n$) since global phases do not affect the Knill-Laflamme conditions~\cite{knill1997theory}. Any Pauli error $E\in\mathcal{P}_n$ belongs to one of three classes related to a stabilizer. 
First, if $E\in S$, then for any codeword $\ket{\psi}$, $E\ket{\psi}=\ket{\psi}$, thus the error acts trivially on the codespace. 
Second, if $E\in C_n(S)\setminus S$, where $C_n(S)$ denotes the centralizer of $S$ in $\mathcal{P}_n$, then $E$ commutes with every $s\in S$, but is not itself in the stabilizer. More precisely, for any $\ket{\psi}\in V_S$ and $s\in S$, $s(E\ket{\psi}) = Es\ket{\psi} = E\ket{\psi}$, implying $E\ket{\psi}\in V_S$. Such an error maps a codeword to another codeword and is therefore undetectable. 
Finally, if $E\in \mathcal{P}_n\setminus C_n(S)$, then $E$ anticommutes with at least one $s\in S$. In this case, $E\ket{\psi}$ lies in a subspace orthogonal to $V_S$, and the error can be detected through syndrome measurement. 
Combining these observations, a set of Pauli errors $\{E_i\}_i\subset \mathcal{P}_n$ is correctable if and only if $E_i^\dag E_j\notin C_n(S)\setminus S$. 

The $n$-qubit phase-free Pauli group is isomorphic to the additive group $\F_2^{2n}$, and hence we can represent the Pauli operators, modulo phase, in binary form through the map $r\colon\mathcal{P}_n\rightarrow\F_2^{2n}$ defined as
\begin{align}\label{eq_rowfunc}
    r\left(\bigotimes_{i=1}^n X^{a_i}Z^{b_i}\right):=[a_1\ \ldots\ a_n\ b_1\ \ldots\ b_n]=[\mathbf{a}\ \mathbf{b}],
\end{align}
where 
$a_i,b_i\in\F_2$ indicate whether $X_i,Z_i$ are present in the associated Pauli string. The weight of any $P\in\mathcal{P}_n$ is the number of qubits with support, namely
\begin{align*}
    w\left(\bigotimes_{i=1}^n X^{\mathbf{a}_i}Z^{\mathbf{b}_i}\right) := \sum_{i=1}^n \mathbbm{1}_{[\mathbf{a}_i=1\ \vee\  \mathbf{b}_i=1]}.
\end{align*}
Accordingly, the distance of a stabilizer code is then defined as
\begin{align}\label{eq_weight}
    d := \underset{E\in C_n(S)\setminus S}{\min} w(E).
\end{align}

A stabilizer that encodes $k$ logical qubits into $n$ physical qubits with distance $d$ is denoted by $[\![n,k,d]\!]$, and is referred to as an $[\![n,k,d]\!]$ stabilizer \cite{SC&QEC}.

A stabilizer may equivalently be specified through its associated stabilizer matrix. Given a stabilizer $S=\langle s_1,\ldots,s_{n-1}\rangle$, the stabilizer matrix is defined as
\begin{equation}\label{eq:stabilizerMatrix}
    G_S := [\mathbf{r(s_1)}^T\ \cdots\ \mathbf{r(s_{n-k})}^T]^T.
\end{equation}
The stabilizer matrix provides a binary representation of the stabilizer and allows its properties to be analysed using symplectic linear algebra. The symplectic inner product on $\F_2^{2n}$ is defined by
\begin{align}\label{eq_symp_form}
    \langle[\mathbf{a}^{(i)}\ \mathbf{b}^{(i)}],[\mathbf{a}^{(j)}\ \mathbf{b}^{(j)}]\rangle_S := \bra{\mathbf{a}^{(i)}}\ket{\mathbf{b}^{(j)}} + \bra{\mathbf{b}^{(i)}}\ket{\mathbf{a}^{(j)}}.
\end{align}
where $\bra{\cdot}\ket{\cdot}$ denotes the standard Euclidean inner product. Two Pauli operators $P,P'\in\mathcal{P}_n$ commute if and only if their corresponding representations $\mathbf{r(P)}, \mathbf{r(P')}$ are symplectically orthogonal. Since $S$ is Abelian by definition, it follows that $\langle \mathbf{r(s_i)},\mathbf{r(s_j)}\rangle_S=0$ for all $s_i,s_j\in S$.

\subsection{Puncturing Stabilizers}\label{sec_puncture}
In the following, we review the theory of puncturing stabilizers developed in \cite{PQSC}. We summarize the puncturing construction and the associated results concerning the parameters of the resulting codes.

Let $S_2:=\{\mathbf{r(s)}\ |\ s\in S\}$ denote the binary image of the stabilizer group under $r$. Since $S$ is a subgroup of $\mathcal{P}_n$, $S_2$ forms a linear subspace of $\F_2^{2n}$. Furthermore, the commutativity of $S$ implies that $S_2$ is symplectically self-orthogonal with respect to \eqref{eq_symp_form}. Let $S_2^{\bot}$ denote the binary image of the centralizer $C_n(S)$ under $r$. Since $S$ is Abelian, it follows that $S_2\subseteq S_2^\bot$. Henceforth, $S_2$ and $S_2^\bot$ will be referred to as the stabilizer subspace and centralizer subspace, respectively.

To define quantum puncturing, first introduce the classical puncturing map $\pi_{I}(\mathbf{v}) := [\mathbf{v}_i]_{i\notin I}$ for some $\mathbf{v}\in\F_2^n$ and an indexing set $I$. This map is then extended to $\F_2^{2n}$ by
\begin{align*}
    \pi^Q_{I}([\mathbf{a}\ \mathbf{b}]) := [\pi_{I}(\mathbf{a})\ \pi_{I}(\mathbf{b})].
\end{align*}
Let $m\in[1;n]$ be an integer and choose $[\alpha\ \beta]\in\F_2^2 \setminus \{[0,0]\}$. The punctured stabilizer subspace is then given by
\begin{equation}\label{eq_puncg2}
    S_2^{[\alpha\ \beta]}\! :=\! \{\pi^Q_{\{m\}}([\mathbf{a}\ \mathbf{b}])\ |\ [\mathbf{a}\ \mathbf{b}]\in S_2,\ \langle [\mathbf{a}_m\ \mathbf{b}_m],[\alpha\ \beta]\rangle_S=0\}.
\end{equation}
Throughout the remainder of this paper, it is assumed that the original stabilizer has distance $d\geq2$. For codes with distance $d=1$, puncturing may remove logical information, and the subsequent results may not hold. The reader is referred to \cite{PQSC} for analysis of such cases. 
 
 The following theorem establishes that $S_2^{[\alpha\ \beta]}$ defines a valid stabilizer subspace.
\begin{theorem}\label{theo_1}
    Let $S_2^{[\alpha\ \beta]}$ be defined by \eqref{eq_puncg2}. Then:
    \begin{enumerate}
        \item $S_2^{[\alpha\ \beta]}$ is symplectically self-orthogonal with respect to \eqref{eq_symp_form}.
        \item $\dim\left(S_2^{[\alpha\ \beta]}\right)=\dim(S_2) - 1$.
        \item $(S_2^{[\alpha\ \beta]})^{\bot} = (S_2^{\bot})^{[\alpha\ \beta]}$.
        \item $\dim((S_2^{\bot})^{[\alpha\ \beta]})=\dim(S_2^{\bot}) - 1$.
    \end{enumerate}
\end{theorem}
For the proof, see \cite{PQSC}. The theorem implies that puncturing preserves symplectic self-orthogonality and therefore yields a valid stabilizer subspace. Moreover, both the stabilizer subspace and the centralizer subspace lose exactly one dimension under puncturing. Consequently, if the original code has the parameters $[\![n,k,d]\!]$, then the punctured code has parameters $[\![n',k',d']\!]$ where
\begin{align*}
    [\![n,k,d]\!]  \xrightarrow{\text{puncturing}} [\![n'=n-1,k'=k,d'\geq d-1]\!].
\end{align*}
The definition \eqref{eq_puncg2} requires evaluating all elements of the stabilizer subspace. For large subspaces, such an approach may be computationally undesirable. It is therefore advantageous to formulate puncturing directly on the stabilizer matrix. To this end, define the extended stabilizer matrix, $G_{\text{ext}}$, as the matrix whose rows form a basis of $S_2^\bot$. The extended stabilizer matrix is obtained by augmenting the stabilizer matrix $G_S$ from \eqref{eq:stabilizerMatrix} with $2k$ additional linearly independent rows that are symplectically orthogonal to all rows of the stabilizer matrix. Suitable row operations on the stabilizer matrix reveals such a basis ~\cite{SC&QEC}. Based on \Cref{theo_1},~\cite{PQSC} derives an efficient puncturing procedure that operates directly on $G_{\text{ext}}$. The resulting algorithm is given in \Cref{alg_punc}. An application of the algorithm performed on a specific stabilizer is presented in \hyperref[app_all]{the appendix}.
\begin{algorithm}[h]
\caption{Puncturing on the stabilizer Matrix}
\label{alg_punc}
\begin{algorithmic}[1]
\item[] \textit{Input:} Centralizer matrix $G_{\text{ext}}\in\F_2^{(n+k)\times 2n}$, $[\alpha\ \beta]\in\F_2^2\setminus\{[0\ 0]\}$ and $m\in[1;n]$

\State Denote row $i$ of $G_\text{ext}$ by $[\mathbf{a}^{(i)}\ \mathbf{b}^{(i)}]$ for $i\in[1;n+k]$

\State Initialize $\ell=0$ and $j=0$

\For{$i={1,\ldots,n-k}$}

    \State Compute $\sigma([\mathbf{a}^{(i)}\ \mathbf{b}^{(i)}]) = \langle[\mathbf{a}^{(i)}_m\ \mathbf{b}^{(i)}_m],[\alpha\ \beta]\rangle_S$

    \If{$\sigma([\mathbf{a}^{(i)}\ \mathbf{b}^{(i)}]) = 1$}
        \State $\ell\gets \ell+1$
        \If{$\ell=1$}
            \State Remove row $i$ from $G_\text{ext}$ and set $j=i$
        \Else 
        \EndIf
    \Else
        \State Continue
    \EndIf

\EndFor
\State Remove column $m$ and $m+n$ from the resulting matrix

\item[] \textit{Output:} Centralizer matrix of stabilizer punctured wrt. $[\alpha\ \beta]$ on qubit $m$
\end{algorithmic}
\end{algorithm}

\section[Methodology]{\textbf{Entanglement-Assisted Stabilizers}}\label{sec_eas}
Entanglement-assisted stabilizers generalize the stabilizer formalism by exploiting pre-shared entanglement between the transmitter and the receiver~\cite{CQEwE}. As discussed in \Cref{sec_preliminaries}, a stabilizer is Abelian, which is equivalent to requiring its associated stabilizer subspace to be symplectically self-orthogonal. This condition imposes a significant restriction on the admissible parity-check matrices. In particular, parity-check matrices of classical linear codes with parameters $[2n,n-k,d]$ can be used to define $[[n,k,d']]$ stabilizers only when their associated subspaces are symplectically self-orthogonal, with $d'$ depending on the particular parity-check matrix. Entanglement assistance overcomes this restriction: A parity-check matrix that fails to be symplectically self-orthogonal corresponds to a non-Abelian Pauli subgroup, which can be made Abelian by embedding it into a larger Hilbert space. To illustrate this idea, consider the two isomorphic non-Abelian subgroups $H,G\leq\mathcal{P}_4$ given by their generators in \eqref{eq_groups}. Here, $H$ is the subgroup from which we want to construct a stabilizer.
\begingroup
\renewcommand{\arraystretch}{1}   
\setlength{\arraycolsep}{3pt} 
\begin{equation}\label{eq_groups}
\begin{array}{ccccccc}
    \begin{NiceArray}{cccccc}
        \Block{1-6}{H}\\\hline
        \widetilde{Z}_1& :=& Z& X& Z& I\\
        \widetilde{X}_1& :=& Z& Z& I& Z\\
        \widetilde{Z}_2& :=& Y& X& X& Z\\
        \widetilde{Z}_3& :=& Z& Y& Y& X
    \end{NiceArray}
     &&&&&&
    \begin{NiceArray}{cccccc}
        \Block{1-6}{G}\\\hline
        Z_1&:=&Z&I&I&I\\
        X_1&:=&X&I&I&I\\
        Z_2&:=&I&Z&I&I\\
        Z_3&:=&I&I&Z&I
    \end{NiceArray}
\end{array}
\end{equation}
\endgroup
Since $H$ and $G$ are isomorphic, they have identical commutation relations, where $Z_1$ and $X_1$ (and similarly $\widetilde{Z}_1$ and $\widetilde{X}_1$) anti-commute while all other pairs commute. The subgroup $H$ can be decomposed into two subgroups based on the commutation relations~\cite{CQEwE}. The subgroup $\text{symp}(H):=\langle \widetilde{Z}_1,\widetilde{X}_1\rangle$ is referred to as the symplectic subgroup of $H$, while $\text{iso}(H):=\langle \widetilde{Z}_2,\widetilde{Z}_3\rangle$ is called the isotropic subgroup. The analogous decomposition exists for $G$. In this example, $\text{symp}(H)$ contains one symplectic pair that must be resolved to make $H$ Abelian. The structure of $G$ provides an evident embedding into $\mathcal{P}_5$ by the addition of a single qubit, yielding the Abelian subgroup $\widehat{G}$. Since $G$ and $H$ are isomorphic, the same extension induces an Abelian embedding $\widehat{H}$. The resulting groups are shown in \eqref{eq_groups_ext}.
\begingroup
\renewcommand{\arraystretch}{1.2}   
\setlength{\arraycolsep}{3pt} 
\begin{equation}\label{eq_groups_ext}
\begin{array}{ccccccc}
    \begin{NiceArray}{cccccc|[tikz=dashed]c}
        \Block{1-7}{\widehat{H}}\\\hline
        \widehat{\widetilde{Z}}_1& :=& Z& X& Z& I& Z\\
        \widehat{\widetilde{X}}_1& :=& Z& Z& I& Z& X\\
        \widehat{\widetilde{Z}}_2& :=& Y& X& X& Z& I\\
        \widehat{\widetilde{Z}}_3& :=& Z& Y& Y& X& I
    \end{NiceArray}
     &&&&&&
    \begin{NiceArray}{cccccc|[tikz=dashed]c}
        \Block{1-7}{\widehat{G}}\\\hline
        \widehat{Z}_1&:=&Z&I&I&I&Z\\
        \widehat{X}_1&:=&X&I&I&I&X\\
        \widehat{Z}_2&:=&I&Z&I&I&I\\
        \widehat{Z}_3&:=&I&I&Z&I&I
    \end{NiceArray}
\end{array}
\end{equation}
\endgroup
The codespace associated with $\widehat{G}$ is
\begin{align*}
    \ket{\psi} =  \ket{\Phi^+ }_{\mathcal{A}_1\mathcal{A}_5}\otimes \ket{00}_{\mathcal{A}_2\mathcal{A}_3}\otimes \ket{\gamma}_{\mathcal{A}_4}
\end{align*}
where $\ket{\Phi^+} = \frac{1}{\sqrt{2}}(\ket{00}+\ket{11})$ is a Bell state, and $\ket{\gamma}_{\mathcal{A}_4}$ is an arbitrary single-qubit state. The codespace of $\widehat{H}$ may be obtained by transforming its stabilizer matrix into the stabilizer matrix of $\widehat{G}$, which simultaneously reveals the corresponding encoding circuit~\cite{QCwE}. 
In general, it can be shown that each symplectic pair requires one shared Bell state in order to obtain an Abelian extension~\cite{optimal_ebits}.

Observe that the generators of $\widehat{H}$ and $\widehat{G}$ act identically on the fifth qubit. Consequently, this qubit is unaffected by the encoding operation and may therefore be distributed to the receiver prior to encoding. More generally, the receiver's halves of the shared Bell pairs need not pass through the noisy communication channel. This additional resource enables entanglement-assisted stabilizers to achieve code parameters that are unattainable within the standard stabilizer framework~\cite{CQEwE}.

Since one qubit from each Bell pair is retained by the receiver, these qubits are assumed to be error-free. As a consequence, the notion of distance differs slightly from that of unassisted stabilizers. Define 
\begin{align*}
    2c:=\abs{\text{symp}(H)} \qquad \text{and} \qquad s:=\abs{\text{iso}(H)}.
\end{align*}
As in the unassisted setting, Pauli errors may be classified according to their commutation relations with the elements in the stabilizer. The undetectable errors are those elements of $C_{n+c}(\widehat{H})\setminus \widehat{H}$ whose components on the receiver's $c$ qubits are equal to the identity. This set corresponds to $C_n(H)\setminus \text{iso}(H)$. The distance of an entanglement-assisted stabilizer is therefore defined as
\begin{align*}
    d := \underset{E\in C_n(H)\setminus \text{iso}(H)}{\min} w(E).
\end{align*}
An entanglement-assisted stabilizer encoding $k$ logical qubits into $n$ physical qubits while consuming $c$ Bell pairs is denoted by $[\![n,k,d;c]\!]$ with $n=s+c+k$. Such a code is referred to as an entanglement-assisted stabilizer with rate given by the tuple $\left(\frac{k}{n},\frac{c}{n}\right)$ \cite{CQEwE}.

A communication system employing entanglement-assisted error-correction is illustrated in \Cref{fig_com_framework}. The communication process consists of two phases. During the setup phase, a source distributes $c$ Bell pairs between the transmitter and the receiver. Subsequently, during the communication phase, the transmitter encodes the logical information jointly with $s$ ancilla qubits and its halves of the shared Bell pairs. The resulting $n=c+s+k$ qubits are transmitted through the noisy channel, while the receiver performs error-correction and decoding using both the received qubits and its locally stored halves of the entangled pairs. 

The setup phase may be performed independently of the communication phase, allowing entanglement resources to be accumulated in advance. The practical viability of this approach depends on maintaining coherence of the stored Bell pairs until decoding. Entanglement purification and memory-protection techniques may be employed to mitigate decoherence during this storage interval \cite{eaqecc_imperfect_ebits}.
\section[results]{\textbf{Puncturing Entanglement-Assisted Stabilizers}}\label{sec_punct_eas}

In this section, we extend the puncturing algorithm developed for unassisted stabilizers to a method capable of reducing the number of required Bell pairs in an entanglement-assisted stabilizer.

\subsection{Entanglement-Assisted Stabilizer as Unassisted Stabilizer}

Consider the $[\![4,1,3;1]\!]$ entanglement-assisted stabilizer $\widehat{H}$ of \eqref{eq_groups_ext}. The stabilizer is obtained by extending the non-Abelian subgroup $H\leq\mathcal{P}_4$ with an additional qubit, thereby obtaining an Abelian stabilizer $\widehat{H}\leq\mathcal{P}_5$. Since the
additional qubit may be interpreted as one half of a pre-shared Bell pair between the transmitter and the receiver, the resulting stabilizer is said to be entanglement-assisted.

An alternative realization is possible in which no entanglement is pre-distributed. For example, consider the stabilizer group $\langle Z_1,Z_2,Z_3,Z_4\rangle\leq\mathcal{P}_5$ with codespace
\begin{align*}
    \ket{0000}_{\mathcal{A}_1\mathcal{A}_2\mathcal{A}_3\mathcal{A}_4}\otimes\ket{\gamma}_{\mathcal{A}_5}.
\end{align*}
By applying a suitable encoding circuit, the codespace of this trivial stabilizer can be mapped into the codespace of $\widehat{H}$. In doing 
so, the setup phase illustrated in \Cref{fig_com_framework} becomes unnecessary. The price paid is that all five physical qubits must be transmitted through the noisy channel, whereas only four channel uses are required in the entanglement-assisted realization.

More generally, every $[\![n,k,d;c]\!]$ entanglement-assisted stabilizer admits an equivalent interpretation as an unassisted stabilizer with parameters $[\![n+c,k,d']\!]$, where the distance $d'$ need not equal the entanglement-assisted distance $d$. 

Let $\widehat{S}\leq\mathcal{P}_{n+c}$ denote an extended subgroup defining an arbitrary entanglement-assisted stabilizer. Denote the first $n$ qubits as channel qubits and the last $c$ qubits as receiver-side qubits. The distance of $\widehat{S}$ as an entanglement-assisted stabilizer is then
\begin{align}\label{eq_ea_dist}
    d := &\underset{E\in C_{n+c}(\widehat{S})\setminus \widehat{S}}{\min} w(E)\\
    &\qquad\text{s.t. } E = \bigotimes_{i=1}^n P_i \otimes \bigotimes_{i=n+1}^{n+c} I, \qquad P_i\in\{I,X,Y,Z\}.
\end{align}
That is, only errors acting on channel qubits are regarded as physically realizable. Since receiver-side qubits are assumed to be noiseless and are never transmitted through the channel, logical errors with support on these qubits do not contribute to the entanglement-assisted distance. 

In contrast, when the same code is interpreted as an unassisted stabilizer on $n+c$ qubits, every qubit is transmitted through the channel, and therefore every physical Pauli error must be considered. The corresponding distance is
\begin{equation}\label{eq_unassisted_dist}
  d' := \underset{E\in C_{n+c}(\widehat{S})\setminus \widehat{S}}{\min} w(E),
\end{equation}
with no restriction on the support of $E$. Since the minimization defining $d'$ is taken over a larger feasible set than that defining $d$, it follows immediately that $d'\leq d$. Consequently, interpreting an entanglement-assisted stabilizer as an unassisted stabilizer can never increase its distance and may reduce it substantially. An example is provided by the stabilizer
\begingroup
\renewcommand{\arraystretch}{1.2}  
\setlength{\arraycolsep}{3pt} 
\begin{equation}\label{eq_group_example}
    \begin{NiceArray}{ccccccc|[tikz=dashed]cc}
        \widehat{\widetilde{Z}}_1&:=&Z&X&Z&I&Z&Z&I\\
        \widehat{\widetilde{X}}_1&:=&Y&Z&X&I&I&X&I\\
        \widehat{\widetilde{Z}}_2&:=&X&I&Y&X&Z&I&Z\\
        \widehat{\widetilde{X}}_2&:=&X&Y&X&Z&X&I&X\\
        \widehat{\widetilde{Z}}_3&:=&Z&X&X&Z&X&I&I\\
        \widehat{\widetilde{Z}}_4&:=&Y&I&I&I&Y&I&I\\
    \end{NiceArray}
\end{equation}
\endgroup
The code described by \eqref{eq_group_example} may be interpreted either as a $[\![5,1,3;2]\!]$ entanglement-assisted stabilizer or as a $[\![7,1,2]\!]$ unassisted stabilizer. Indeed, the minimum-weight logical errors (for example $Z_4X_7$) have weight two and support on the last two qubits, corresponding to the receiver-side halves of the shared Bell pairs. These errors are excluded from the minimization in \eqref{eq_ea_dist} but included in \eqref{eq_unassisted_dist}. Hence, the entanglement-assisted distance is $d=3$, whereas the unassisted distance is $d'=2$.

\subsection{Extending Pre-Established Framework}\label{subsec_ext_framework}

Since every $[\![n,k,d;c]\!]$ entanglement-assisted stabilizer can be represented as an unassisted $[\![n+c,k,d']\!]$ stabilizer, this conversion provides an extension of \Cref{sec_puncture} to the entanglement-assisted setting. Rather than puncturing the entanglement-assisted stabilizer directly, we first convert it into an equivalent unassisted stabilizer, apply the puncturing algorithm, and finally recover an entanglement-assisted representation of the punctured code. The procedure is illustrated in \Cref{fig_distance}.

\begin{figure}[h]
    \centering
    \centering
\resizebox{0.7\linewidth}{!}{
\begin{tikzpicture}[
main/.style={rectangle, draw, fill=white, 
    text width=3cm, text centered, minimum height=1cm},
main2/.style={circle, draw, fill=white, 
    text width=.1mm, text centered, minimum height=.1mm},
line/.style = {draw, rounded corners, thick},
dot/.style={circle, color=black, fill=black, inner sep=0pt, minimum size=5pt},
linetext/.style = {font=\footnotesize, align=center}
]
\coordinate (ea1) at (0,3);
\coordinate (stab1) at (0,0);
\coordinate (stab2) at (6,0);
\coordinate (ea2) at (6,3);

\node[main] (ea1) at (ea1){$[[n,k,d;c]]$};
\node[main] (stab1) at (stab1){$[[n+c,k,d']]$};
\node[main] (stab2) at (stab2){$[[n+c-1,k,d'_p]]$};
\node[main] (ea2) at (ea2){$[[n,k,d_p;c-1]]$};
\path[line,->](ea1.south) -- node[pos=0.5,left]{Transform}(stab1.north);
\path[line,->](stab1.east) -- node[pos=0.5,above]{Puncture}(stab2.west);
\path[line,->](stab2.north) -- node[pos=0.5,right]{Transform}(ea2.south);

\end{tikzpicture}
}
    \caption{The four steps of puncturing an entanglement-assisted stabilizer.}
    \label{fig_distance}
\end{figure}
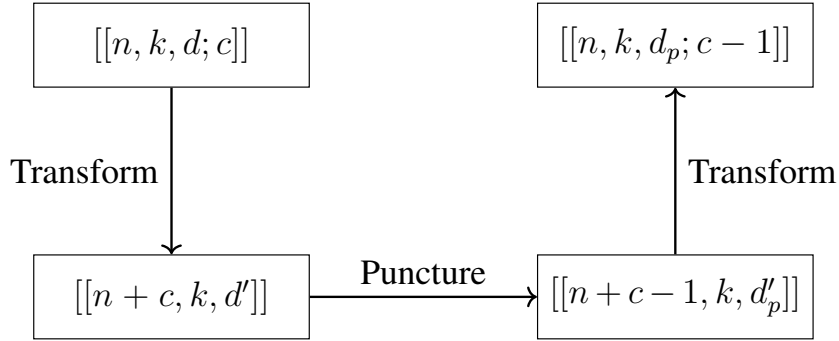

The only non-trivial step remaining is showing that the punctured unassisted stabilizer can always be mapped back to an entanglement-assisted stabilizer. Let 
\begin{align}
    \widehat{S}=\langle \widehat{\widetilde{Z}}_1,\widehat{\widetilde{X}}_1,\ldots,\widehat{\widetilde{Z}}_c,\widehat{\widetilde{X}}_c,\widehat{\widetilde{Z}}_{c+1},\ldots,\widehat{\widetilde{Z}}_{c+s}\rangle\leq \mathcal{P}_{n+c}
\end{align}
define a $[\![n,k,d;c]\!]$ entanglement-assisted stabilizer. By construction, each receiver-side qubit $n+m$ for $m\in[1;c]$ satisfies
\begin{align*}
    (\widehat{\widetilde{Z}}_m)_{n+m} = Z, \qquad (\widehat{\widetilde{X}}_m)_{n+m} = X,
\end{align*}
while for $i\in[1;c+s]\setminus\{m\}$ and $j\in[1;c]\setminus\{m\}$
\begin{align*}
    (\widehat{\widetilde{Z}}_i)_{n+m} = (\widehat{\widetilde{X}}_j)_{n+m} = I.
\end{align*}
Hence, the $(n+m)$'th qubit is occupied exclusively by the symplectic pair $(\widehat{\widetilde Z}_m,\widehat{\widetilde X}_m)$.

Consider puncturing qubit $n+m$. Since the only generators that acts non-trivially on the qubit is the symplectic pair $(\widehat{\widetilde{Z}}_m,\widehat{\widetilde{X}}_m)$, the puncturing algorithm only affects these generators. The algorithm removes exactly one element of the symplectic pair together with the punctured qubit, while the other element survives, now as an isotropic generator. Denote the generators after puncturing by $\{\widehat{\widetilde{Z}}^p_1,\widehat{\widetilde{X}}^p_1,\ldots,\widehat{\widetilde{Z}}^p_{c-1},\widehat{\widetilde{X}}^p_{c-1},\widehat{\widetilde{Z}}^p_{c},\ldots,\widehat{\widetilde{Z}}^p_{c+s}\}$, and define $\widetilde{\pi}_I:\mathcal{P}_n\rightarrow\mathcal{P}_{n-\abs{I}}$ by
\begin{align*}
    \widetilde{\pi}_I\left(\bigotimes_{i=1}^n X^{\mathbf{a}_i}Z^{\mathbf{b}_i}\right) := \bigotimes_{i\notin I} X^{\mathbf{a}_i}Z^{\mathbf{b}_i},
\end{align*}
for some indexing set $I\subset[1;n]$. Then, the surviving half of the affected symplectic pair, now part of the isotropic group, is given by
\begin{align}\label{eq_trans_punct}
    \widehat{\widetilde{Z}}^p_c = \begin{cases}
        \widetilde{\pi}_{\{m\}}(\widehat{\widetilde{Z}}_m), \qquad &[\alpha\ \beta]=[0\ 1],\\
        \widetilde{\pi}_{\{m\}}(\widehat{\widetilde{X}}_m), \qquad &[\alpha\ \beta]=[1\ 0],\\
        \widetilde{\pi}_{\{m\}}(\widehat{\widetilde{Z}}_m) + \widetilde{\pi}_{\{m\}}(\widehat{\widetilde{X}}_m), \qquad &[\alpha\ \beta]=[1\ 1].
    \end{cases}
\end{align}
The unaffected symplectic pairs are merely relabeled according to
\begin{align*}
    \widehat{\widetilde{A}}^p_i = \begin{cases}
        \widetilde{\pi}_{\{m\}}(\widehat{\widetilde{A}}_i), \qquad &1\leq i<m,\\
        \widetilde{\pi}_{\{m\}}(\widehat{\widetilde{A}}_{i+1}), \qquad &m\leq i< c
    \end{cases}, \qquad A\in\{Z,X\}.
\end{align*}
while 
\begin{align}
    \widehat{\widetilde{Z}}^p_{i}=\widetilde{\pi}_{\{m\}}(\widehat{\widetilde{Z}}_{i}), \qquad c<i\leq c+s. 
\end{align}    
Since $\widehat{\widetilde Z}_c^p$ acts trivially on all remaining receiver-side qubits, it commutes with every generator of the punctured stabilizer. Consequently, after deleting the remaining $c-1$ receiver-side qubits, the punctured stabilizer possesses exactly one additional isotropic generator and only $c-1$ symplectic pairs. Applying the inverse symplectic embedding therefore introduces precisely $c-1$ receiver-side qubits, yielding a valid entanglement-assisted stabilizer representation.

Let $\widehat{S}\leq\mathcal{P}_{n+c}$ define an $[\![n,k,d;c]\!]$ entanglement-assisted stabilizer. Puncturing one receiver-side qubit through the induced puncturing procedure produces an entanglement-assisted stabilizer defined by $\widehat{S}_p\leq\mathcal{P}_{n+c-1}$ with parameters $[\![n_p=n,k_p=k,d_p;c_p=c-1]\!]$. The number of transmitted qubits and logical qubits is preserved, while the required number of Bell pairs is reduced by one. The parameter transformation can therefore be summarized as
\begin{align*}
    [\![n,k,d;c]\!]  \xrightarrow{\text{puncturing}} [\![n_p=n,k_p=k,d_p;c_p=c-1]\!].
\end{align*}

Since one symplectic pair is replaced by an isotropic generator, the qubit previously associated with one half of a Bell pair on the transmitter side becomes an ancilla qubit. Consequently, $n_p=k+(s+1)+(c-1)=k+s+c=n$ showing that the total number of transmitted qubits remains unchanged. 

Unlike the parameters $n$, $k$, and $c$, the distance is highly volatile and far less predictable. Three distinct mechanisms may affect the distance: 
\begin{enumerate} 
    \item the conversion $ [\![n,k,d;c]\!] \rightarrow [\![n+c,k,d']\!] $ may decrease the distance from $d$ to $d'$, 
    \item puncturing may reduce the distance by at most one, yielding $ d'_p\ge d'-1 $, 
    \item the inverse conversion may increase the distance by excluding logical operators supported on receiver-side qubits. 
\end{enumerate}
Combining these observations yields $d_p\geq d'-1$. Hence $d-d_p\leq (d-d')+1$. Defining $B_d := (d-d')+1$, we obtain the bound $d_p\geq d-B_d$. Thus the total loss in distance caused by the puncturing procedure is bounded by $B_d$. For the stabilizer of \eqref{eq_group_example}, the corresponding unassisted code has distance $d'=2$ while the entanglement-assisted code has distance $d=3$. Therefore, $B_d=(3-2)+1=2$. This worst-case loss is attained when puncturing either the first or the last Bell pair with respect to the $X$ operator.

\subsection{Distance Considerations}\label{subsec_dist_cons}

For unassisted stabilizers, puncturing decreases the distance by at most one. More precisely,~\cite{PQSC} showed that if the $m$'th qubit of an $[\![n,k,d']\!]$ stabilizer is punctured with respect to $[\alpha\ \beta]\in\F_2^2\setminus\{[0\ 0]\}$ and $[\mathbf{
a}_m\ \mathbf{b}_m]\neq[\alpha\ \beta]$ for every $[\mathbf{a}\ \mathbf{b}]\in S_2^\bot\setminus S_2$ with $\sum_{i=1}^n \mathbbm{1}_{[\mathbf{a}_i=1\ \vee\  \mathbf{b}_i=1]}=d'$,  then the punctured code has distance $d'_p\geq d'$. For entanglement-assisted stabilizers the distance behaves less predictably, since the puncturing procedure consists of three operations: A conversion into an unassisted stabilizer, puncturing, and an inverse conversion. We can still provide requirements to ensure the distance is preserved; however, as the next theorem will show, they are quite extensive. 
\begin{theorem}\label{theo_2}
    Let $\widehat{S}\leq\mathcal{P}_{n+c}$ define an $[\![n,k,d;c]\!]$ entanglement-assisted stabilizer. Consider puncturing the receiver-side qubit $m\in[n+1;n+c]$ with respect to $[\alpha\ \beta]\in\F_2^2\setminus\{[0\ 0]\}$. Suppose every logical error $E\in C_{n+c}(\widehat{S})\setminus \widehat{S}$ satisfying $w(E)\leq d$ obeys at least one of the following conditions: \begin{enumerate} 
        \item $E$ has support on some receiver-side qubit $i\in[n+1;n+c]\setminus\{m\}$, or 
        \item $E$ has no support in $[n+1;n+c]\setminus\{m\}$ and $\mathbf{r(E)}_{m,n+m}\neq[\alpha\ \beta]$.
    \end{enumerate} 
    Then the entanglement-assisted stabilizer obtained after puncturing and performing the inverse conversion has distance $d_p\ge d$.  
\end{theorem}
\begin{proof} 
    Let $E\in C_{n+c}(\widehat{S})\setminus \widehat{S}$ with $w(E)\leq d$. To prove $d_p\geq d$, it suffices to show that every such $E$ is absent from the logical error set of the resulting entanglement-assisted code. To this end, we divide the set 
    \begin{align}
        A:=\{E\ |\ E\in C_{n+c}(\widehat{S})\setminus \widehat{S}\wedge w(E)\leq d\}
    \end{align}
    into two disjoint subsets; 
    \begin{align}
        B_1&:=\{E\ |\ E\in A\wedge\exists i\in[n+1;n+c]\setminus\{m\}:\ E_i\neq I\},\\
        B_2&:=\{E\ |\ E\in A\wedge\forall i\in[n+1;n+c]\setminus\{m\}:\ E_i=I\}.
    \end{align}
    If $E\in B_1$, then $E$ has support on at least one receiver-side qubit distinct from the punctured qubit. After puncturing, the inverse conversion removes all logical errors whose support extends to one of the remaining receiver-side qubits. Hence $E$ cannot contribute to the distance of the final entanglement-assisted code. This is the requirement of Condition 1.
    
    Conversely, if $E\in B_2$, then $E$ may survive the inverse conversion, since its support is contained entirely within the channel qubits and the punctured receiver-side qubit. Therefore it must be eliminated by the puncturing operation itself or be of weight $d$ after puncturing. By the puncturing criterion of \Cref{sec_puncture}, a logical operator is removed precisely when $\langle \mathbf{r(E)}_{m,n+m},[\alpha\ \beta]\rangle_S = 1$. Since $E$ has support on qubit $m$ (unless possibly if $w(E)=d$. If $E$ does not have support on qubit $m$ puncturing does not alter the weight of $E$ and it may remain after puncturing), this condition is equivalent to $\mathbf{r(E)}_{m,n+m}\neq[\alpha\ \beta]$. Consequently every operator satisfying Condition 2 is either removed during puncturing or of weight $\geq d$ after puncturing. 
    
    In conclusion, every logical error of weight $\leq d$ satisfying condition 1 or 2 is either 
    \begin{enumerate} 
        \item removed by puncturing, 
        \item excluded by the inverse entanglement-assisted conversion, or
        \item of weight $d$ after puncturing.
    \end{enumerate} 
    Hence, if all $E\in A$ satisfy either of the conditions, no logical error of weight smaller than $d$ remains in the final code, proving $d_p\geq d$. 
\end{proof} 

As a side-note; from \eqref{eq_puncg2} and the structure of the inverse conversion, it is evident that all logical errors $\{E\  |\ E\in C_{n+c}(\widehat{S})\setminus \widehat{S}\,\wedge\, \forall i\in[n+1;n+c]:\ E_i=I\}$ of the unpunctured entanglement-assisted stabilizer remain after the puncturing and the inverse conversion. Consequently, $d_p\leq d$ for any puncturing scenario. 

The theorem provides a condition for preserving the entanglement-assisted distance during puncturing. In practice, only the receiver-side qubits are candidates for puncturing, and each such qubit admits three non-trivial puncturing choices corresponding to $[\alpha\ \beta]\in\{[1\ 0],[0\ 1],[1\ 1]\}$. Consequently, an $[\![n,k,d;c]\!]$ entanglement-assisted stabilizer has exactly $3c$ distinct puncturing scenarios. For codes with a modest number of Bell pairs, it may therefore be computationally advantageous to evaluate the distance of all $3c$ punctured codes offline and store the optimal puncturing choice. When a reduction in the required entanglement becomes necessary, the puncturing procedure can then be performed by selecting the pre-computed puncturing pattern, yielding the largest distance.

\subsection{Statistical Analysis}\label{subsec_stat} 
The previous sections established an upper bound on the loss in distance under entanglement-assisted puncturing and provided sufficient conditions under which the distance is preserved. To assess the practical relevance of these results, we now investigate the statistical behavior of the distance under puncturing. Following the sequence of transformations illustrated in \Cref{fig_distance}, define 
\begin{align*} 
    \Delta_1 &:= d'-d,\\ 
    \Delta_2 &:= d'_p-d',\\ 
    \Delta_3 &:= d_p-d'_p,\\ 
    \Delta_4 &:= d_p-d. 
\end{align*} 
Thus, negative values correspond to a decrease in distance, whereas positive values correspond to an increase.
Notice that the most important of these values is $\Delta_4$ since it describes the total distance decrease of the entire puncturing procedure.
Recall from \Cref{subsec_ext_framework} that the total loss in distance is bounded by $d-d_p\leq B_d=1-\Delta_1$. To quantify the tightness of this bound, define
\begin{align}
    \Delta_p = B_d-(d-d_p) = 1-\Delta_1+\Delta_4. 
\end{align}
Hence $\Delta_p$ measures the gap between the theoretical upper bound on the distance loss and the loss actually observed after puncturing. A value of $\Delta_p=0$ indicates that the bound is attained.

To assess the typical values of $\Delta_p$, we conducted a numerical experiment, where $10\,000$ entanglement-assisted stabilizers were generated randomly. The parameters were chosen such that $n\in[5;8]$, $k\in[1;5]$ and $c\in[1;5]$, with $n=s+c+k$ and $\ell=2c+s\in[n-2;n+2]$. For each realization, $\ell$ Pauli operators were sampled uniformly from $\mathcal P_n$. Whenever the resulting subgroup was non-Abelian and the associated unassisted stabilizer possessed distance at least two, the code was accepted. Repeating this procedure and then discarding any repeated codes yielded $9\,999$ distinct codes that were considered in the study. For each code, every admissible puncturing operation on every receiver-side Bell-pair qubit was evaluated. This resulted in a total of $79\,437$ puncturing scenarios. For each scenario the quantities $\Delta_i$ for $i\in\{1,2,3,4,p\}$ and $B_d$ were computed. The resulting empirical distributions are shown in \Cref{fig_delta}.
\begin{figure*}[tb] 
    \centering 
    \includegraphics[width=\linewidth]{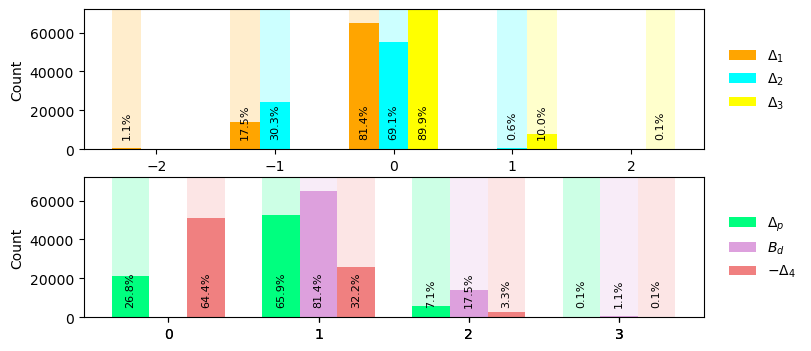} 
    \caption{Histograms of $\Delta_i$ for $i\in\{1,2,3,4,p\}$ and $B_d$ obtained from the $79\,437$ puncturing scenarios considered in this section.} 
    \label{fig_delta} 
\end{figure*} 

The histogram of $\Delta_1$ demonstrates that in $81.4\%$ of all puncturing scenarios the distance was unchanged when the entanglement-assisted code was viewed as an unassisted stabilizer. Similarly, the histogram of $\Delta_2$ shows that $69.1\%$ of the puncturing operations preserved the distance of the unassisted code. The distributions of $\Delta_1$ and $\Delta_3$ are strongly concentrated around zero. Only a small fraction of the examined codes exhibited a distance change greater than one when transitioning between the entanglement-assisted and unassisted representations. Interestingly, the histogram of $\Delta_2$ contains a small positive tail, indicating that approximately $0.6\%$ of puncturing operations increased the distance of the unassisted stabilizer. This behavior is consistent with puncturing removing low-weight logical errors from the centralizer.

Considering the overall puncturing procedure, represented by $\Delta_4$, less than $4\%$ of the puncturing scenarios experienced a distance reduction greater than one, while $64.4\%$ experienced no reduction in distance at all. Thus, for the majority of the examined codes, puncturing a Bell pair incurred either no loss or only a modest loss in distance. The behavior of the bound derived in \Cref{subsec_ext_framework} may be assessed through $\Delta_p$. Although the bound satisfied $B_d\geq 2$ for more than $18\%$ of all puncturing scenarios, equality was attained in only $26.8\%$ of the cases. Consequently, the bound is generally conservative and frequently overestimates the actual loss in distance.

Finally, among the $9\,999$ generated codes, $14.75\%$ exhibited a strictly larger distance in their entanglement-assisted representation than in their associated unassisted representation. Of these codes, $81.4\%$ admitted at least one puncturing scenario for which the original entanglement-assisted distance was preserved. This observation suggests that, for a substantial fraction of codes that gain distance from entanglement assistance, it remains possible to reduce the required number of Bell pairs without sacrificing error-correcting performance.
\section{\textbf{Discussion}}\label{sec_discussion}
The practical benefit of entanglement-assisted stabilizers depends on both the errors affecting the transmitted qubits and the reliability of the pre-shared Bell pairs. For an $[\![n,k,d;c]\!]$ entanglement-assisted stabilizer, the minimum distance $d$ characterizes protection against errors on the $n$ transmitted qubits under the assumption that the receiver-side halves of the Bell pairs are noiseless. The validity of this idealization depends on the noise accumulated throughout entanglement distribution and subsequent storage. If receiver-side errors become appreciable, the minimum distance $d$ alone is insufficient to characterize the performance under the combined noise processes. 
This issue becomes particularly relevant when Bell pairs are generated at different times. Differences in storage duration can result in unequal fidelities at the time of decoding, even when the pairs have comparable initial quality and experience similar memory noise. An unreliable Bell pair can therefore be more detrimental than a pair whose unavailability is known before encoding; a known resource shortage allows the code to be adapted, whereas an unknown error on a receiver-side qubit falls outside the noise model. 
The puncturing framework developed in this paper suggests a means of adapting the code to both the availability and the quality of the shared entanglement. Before encoding, a pair with sufficiently low fidelity could be excluded by puncturing its receiver-side qubit and re-initializing the corresponding transmitter-side qubit as a local ancilla. The resulting code has parameters $[\![n,k,d_p;c-1]\!]$, preserving the numbers of transmitted and logical qubits while reducing the entanglement requirement. 
Such an application introduces a tradeoff between the reliability of the retained entanglement and the distance of the punctured code. When the distance-preservation conditions established in \cref{theo_2} hold, an unreliable pair can be excluded without reducing the minimum distance. More generally, however, selecting a pair for removal should account for both its quality and the consequences of puncturing that qubit. Establishing when this strategy improves logical error rates requires evaluating the original and punctured codes under a joint model of memory and channel noise. 

\section{\textbf{Conclusion}}\label{sec_conclusion} 
This paper extends the puncturing framework of \cite{PQSC} to encompass the receiver-side Bell pairs of entanglement-assisted stabilizers. The proposed approach consists of three steps. First, an $[\![n,k,d;c]\!]$ entanglement-assisted stabilizer is interpreted as an associated $[\![n+c,k,d']\!]$ unassisted stabilizer. Second, the unassisted stabilizer is punctured using existing stabilizer puncturing techniques. Finally, the punctured code is mapped back to an entanglement-assisted stabilizer, yielding an $[\![n,k,d_p;c-1]\!]$ code. Consequently, the procedure reduces the entanglement consumption by one Bell pair while preserving the number of logical qubits. 

In contrast to the parameters $n$, $k$, and $c$, whose evolution under the puncturing procedure is straightforward, the behavior of the distance is substantially more intricate. The primary challenge stems from the fact that entanglement-assisted stabilizers may exhibit a larger distance than their associated unassisted stabilizers. This advantage arises because the receiver-side halves of the Bell pairs are assumed noiseless and are therefore excluded from the set of admissible errors when evaluating the distance. As a result, puncturing performed through the unassisted representation can potentially eliminate this advantage and lead to a degradation in distance when the code is transformed back into its entanglement-assisted form. To address this issue, a condition was derived under which the distance advantage of the entanglement-assisted code is preserved throughout the puncturing procedure.

Furthermore, the statistical study presented in \Cref{subsec_stat} indicates that randomly generated entanglement-assisted stabilizers rarely possess a distance advantage over their associated unassisted stabilizers. In these cases, the puncturing procedure exhibits considerably more predictable behavior, as the distance is less sensitive to the intermediate transformation into an unassisted code. Among the randomly generated codes that did exhibit a strict distance advantage, $81.4\%$ possessed at least one puncturing configuration for which the original distance was maintained. 

\section{\textbf{Future Work}}\label{sec_future_work}
Several directions remain for future work. First, the shared Bell pairs have been assumed to be perfect throughout this work. Although  entanglement purification and additional error-correction procedures may mitigate imperfections, the effect of using Bell pairs with fidelity $F<1$ remains to be investigated. Second, the proposed puncturing method applies only to the receiver's halves of the shared Bell pairs. Extending the method to allow puncturing of arbitrary qubits would broaden its applicability. Third, it would be interesting to see how puncturing relates to catalytic EAQECC~\cite{brun2014catalytic}, which address a related problem of limited entanglement availability. Finally, the potential use of puncturing for adaptive error-correction depends on efficiently updating the encoding and decoding circuits. In several examples, the encoder circuit obtained after puncturing retains much of the structure of the original circuit. However, a general characterization of this transformation and the corresponding updates to the decoder remains to be developed.
\section*{\textbf{Acknowledgments}}
This work was supported, in part, by the Danish National Research Foundation (DNRF), through the Center CLASSIQUE, grant nr.\ 187.

\bibliography{litteratur}

\appendix
\phantomsection
\label{app_all}
This appendix illustrates the puncturing procedure in \cref{alg_punc}
for the stabilizer defined in \eqref{eq_group_example}. We write $N=7$
for the total number of qubits in its unassisted representation; thus $N=n+c$ if we consider the entanglement-assisted setting. The
first five qubits belong to the transmitter, while qubits $6$ and $7$
are the receiver's halves of the two shared Bell pairs.

The extended stabilizer matrix is
\begingroup
\renewcommand{\arraystretch}{1.2}   
\setlength{\arraycolsep}{3pt} 
\begin{equation}\label{eq_punct_example_ext}
    \begin{NiceArray}{c|ccccc|[tikz=dashed]cc|ccccc|[tikz=dashed]cc}
        [\mathbf{a}^{(1)}\ \mathbf{b}^{(1)}]&0&1&0&0&0&0&0& 1&0&1&0&1&1&0\\
        [\mathbf{a}^{(2)}\ \mathbf{b}^{(2)}]&1&0&1&0&0&1&0& 1&1&0&0&0&0&0\\
        [\mathbf{a}^{(3)}\ \mathbf{b}^{(3)}]&1&0&1&1&0&0&0& 0&0&1&0&1&0&1\\
        [\mathbf{a}^{(4)}\ \mathbf{b}^{(4)}]&1&1&1&0&1&0&1& 0&1&0&1&0&0&0\\
        [\mathbf{a}^{(5)}\ \mathbf{b}^{(5)}]&0&1&1&0&1&0&0& 1&0&0&1&0&0&0\\
        [\mathbf{a}^{(6)}\ \mathbf{b}^{(6)}]&1&0&0&0&1&0&0& 1&0&0&0&1&0&0\\
        \hline
        [\mathbf{a}^{(7)}\ \mathbf{b}^{(7)}]&0&0&0&0&0&0&1& 0&0&0&1&0&0&0\\
        [\mathbf{a}^{(8)}\ \mathbf{b}^{(8)}]&0&0&0&0&0&0&0& 1&0&1&0&1&0&1
    \end{NiceArray}
\end{equation}
\endgroup
The first six rows form a basis of the stabilizer subspace $S_2$, and all eight rows form a basis of its symplectic dual $S_2^\perp$. The horizontal line separates the stabilizer basis from the two additional centralizer basis vectors, while the dashed vertical lines separate the transmitter and receiver coordinates within each binary block. The associated unassisted stabilizer has parameters $[\![7,1,2]\!]$.

We puncture the receiver qubit $m=6$ with respect to $[\alpha\ \beta]=[1\ 1]$. Define the linear functional
\begin{equation*}
    \sigma([\mathbf{a}\ \mathbf{b}])
    :=\langle[\mathbf{a}_6\ \mathbf{b}_6],[1\ 1]\rangle_S
    =\mathbf{a}_6+\mathbf{b}_6,
\end{equation*}
where addition is done over $\F_2$. Thus, a vector satisfies the puncturing condition precisely when its entries in columns $6$ and $N+6=13$ are either $[0\ 0]$ or $[1\ 1]$. For the rows of \eqref{eq_punct_example_ext},
\begin{equation*}
    \sigma([\mathbf{a}^{(1)}\ \mathbf{b}^{(1)}])=\sigma([\mathbf{a}^{(2)}\ \mathbf{b}^{(2)}])=1,
    \qquad
    \sigma([\mathbf{a}^{(i)}\ \mathbf{b}^{(i)}])=0\quad\text{for }i=3,\ldots,8.
\end{equation*}
Accordingly, $[\mathbf{a}^{(1)}\ \mathbf{b}^{(1)}]$ is selected as the pivot row. The only required row operation is
\begin{equation*}
    \widetilde{[\mathbf{a}^{(2)}\ \mathbf{b}^{(2)}]}=[\mathbf{a}^{(2)}\ \mathbf{b}^{(2)}]+[\mathbf{a}^{(1)}\ \mathbf{b}^{(1)}].
\end{equation*}
The transformed matrix is
\begingroup
\renewcommand{\arraystretch}{1.2}   
\setlength{\arraycolsep}{3pt} 
\begin{equation}
    \begin{NiceArray}{c|ccccc|[tikz=dashed]cc|ccccc|[tikz=dashed]cc}
    \CodeBefore
        \columncolor{gray!20}{7}
        \columncolor{gray!20}{14} 
        \rowcolor{gray!20}{1}    
    \Body
        [\mathbf{a}^{(1)}\ \mathbf{b}^{(1)}]&0&1&0&0&0&0&0& 1&0&1&0&1&1&0\\
        \widetilde{[\mathbf{a}^{(2)}\ \mathbf{b}^{(2)}]}&1&1&1&0&0&1&0& 0&1&1&0&1&1&0\\
        [\mathbf{a}^{(3)}\ \mathbf{b}^{(3)}]&1&0&1&1&0&0&0& 0&0&1&0&1&0&1\\
        [\mathbf{a}^{(4)}\ \mathbf{b}^{(4)}]&1&1&1&0&1&0&1& 0&1&0&1&0&0&0\\
        [\mathbf{a}^{(5)}\ \mathbf{b}^{(5)}]&0&1&1&0&1&0&0& 1&0&0&1&0&0&0\\
        [\mathbf{a}^{(6)}\ \mathbf{b}^{(6)}]&1&0&0&0&1&0&0& 1&0&0&0&1&0&0\\
        \hline
        [\mathbf{a}^{(7)}\ \mathbf{b}^{(7)}]&0&0&0&0&0&0&1& 0&0&0&1&0&0&0\\
        [\mathbf{a}^{(8)}\ \mathbf{b}^{(8)}]&0&0&0&0&0&0&0& 1&0&1&0&1&0&1
    \end{NiceArray}
\end{equation}
\endgroup
Deleting the shaded pivot row and columns $6$ and $13$ implements the puncturing map. In particular,
\begin{equation*}
    S_2^{[1\ 1]}=\operatorname{span}_{\F_2}\left\{\pi^Q_{\{6\}}\left(\widetilde{[\mathbf{a}^{(2)}\ \mathbf{b}^{(2)}]}\right),\pi^Q_{\{6\}}\left([\mathbf{a}^{(3)}\ \mathbf{b}^{(3)}]\right),\ldots,\pi^Q_{\{6\}}\left([\mathbf{a}^{(6)}\ \mathbf{b}^{(6)}]\right)\right\}.
\end{equation*}
The seven retained and punctured rows form a basis
of $S_2^{[1\ 1]}$. Written as elements of $\mathcal{P}_N$, they are
\begingroup
\renewcommand{\arraystretch}{1}  
\setlength{\arraycolsep}{3pt} 
\begin{equation}
    \begin{NiceArray}{ccccc|[tikz=dashed]cc}
        X&Y&Y&I&Z&I\\
        X&I&Y&X&Z&Z\\
        X&Y&X&Z&X&X\\
        Z&X&X&Z&X&I\\
        Y&I&I&I&Y&I\\
        \hline
        I&I&I&Z&I&X\\
        Z&I&Z&I&Z&Z
    \end{NiceArray}
\end{equation}
\endgroup
The first five rows generate the punctured stabilizer $S^{[1\ 1]}$. The last two rows represent logical errors whose binary representation completes the basis of $(S_2^{[1\ 1]})^\perp$. The final column corresponds to the remaining receiver qubit, originally qubit $7$.

The punctured unassisted stabilizer has parameters $[\![6,1,2]\!]$, whereas the punctured entanglement-assisted stabilizer has parameters $[\![5,1,2;1]\!]$. The minimum weight logical errors of the punctured unassisted stabilizer is $IIIZIX$, $IIZIYI$, $IIYXII$ and $YIZIII$, explaining why the punctured entanglement-assisted stabilizer achieves no advantage in distance compared to the punctured unassisted stabilizer.

\end{document}